\documentclass[smallcondensed]{svjour3}     
\usepackage{graphicx}
\usepackage[all]{xy}
\usepackage{amsmath}
\usepackage{hyperref}
\usepackage{amssymb}
\usepackage{nicefrac}
\def\R{\mathbb{R}}
\def\N{\mathbb{N}}
\def\Z{\mathbb{Z}}

\def\Bool{I\!\!B}

\begin{document}

\title{Some notes on the abstraction operation for Multi-Terminal Binary Decision Diagrams\footnote{The final publication is available at \href{link.springer.com}{link.springer.com}}
}

\titlerunning{Notes on the abstraction operation for MTBDDs}        


\author{Ludwig Griebl        \and
        Johann Schuster 
}



\institute{L. Griebl \at
              University of applied sciences Landshut \\
              \email{griebl@fh-landshut.de}           
           \and
           J. Schuster \at
           University of the Federal Armed Forces Munich \\
              \email{johann.schuster@unibw.de}
}

\date{}

\maketitle

\begin{abstract}
The starting point of this work are inaccurate 
statements found in the literature for Multi-terminal Binary Decision Diagrams (MTBDDs) \cite{bahar:97,kuntz:06,Siegle:Habil}
regarding the well-definedness of the MTBDD abstraction operation.
The statements try to relate an operation $*$ on a set of terminal values $M$ to the 
property that the abstraction over this operation does depend on the order of the abstracted variables.
This paper gives a necessary and sufficient condition for the independence of the abstraction operation
of the order of the abstracted variables in the case of an underlying monoid and it treats 
the more general setting of a magma.
\end{abstract}

\section{Introduction}

Multi-terminal Binary Decision Diagrams (MTBDDs) are widely used to compactly store functions of Boolean sets.
In the model-checking context they are often used to store generator matrices of large Markov chains
generated in a compositional way \cite{parker:prism}.
This work considers an elementary operation on MTBDDs (sometimes also called Algebraic Decision Diagrams), which is called ``Abstraction''. 
Whenever a set $S$, $|S|> 1$ of MTBDD-variables is to be abstracted,
this operation is not necessarily well-defined, as it can depend on the order of abstracted variables.
We try to formulate as precise as possible the criteria for a well-defined abstraction over more than one variable.

We will call the algebraic structure that leads to well-defined abstraction operations ``abstractable magmas'', but
we would like to add that these magmas have a number of names in the literature as e.g.~
``medial'', ``abelian'', ``alternation'', ``transposition'', ``interchange'', ``bi-commutative'', ``bisymmetric'', ``surcommutative'', ``entropic''
\cite{strecker:74,jezek:83}. 
Whenever left- and right-division is possible (i.e.~the magma is already a quasigroup), some deep structure theorems are known for the medial case.
For a recent survey we refer to \cite{shcher:05}.

So far, there have been only statements on well-defined abstraction operations for MTBDDs in the literature that did not 
use abstractable magmas \cite{bahar:97,kuntz:06,Siegle:Habil}. We show in this paper that all the aforementioned statements so far are inaccurate and show that
the correct criterion are abstractable magmas.

The paper is organised as follows: In Sec.~\ref{sec:general} we give some basic definitions.
The main statements and lemmata are formulated
in Sec.~\ref{sec:twoelementabstraction} and Sec.~\ref{sec:bign}. 
Sec.~\ref{sec:assoc} gives the connection to associativity and
commutativity. Some illustrating examples are given in Sec.~\ref{sec:monoid}-\ref{sec:semigroup}, 
Sec.~\ref{sec:onR} transfers the results to operations on $\R$ (modulo some euqivalence class) and Sec.~\ref{sec:conclusion} concludes the paper.

\section{Definitions}
\label{sec:general}

\begin{definition}[Algebraic structures]
A \emph{magma} $(M,*)$ consists of a set $M$ equipped with a single binary operation $*: M \times M \rightarrow M$. 
A binary operation is \emph{closed} by definition (i.e.~$a*b\in M$ for all $a,b\in M$), but no other axioms are imposed on the operation.
A \emph{semigroup} $(M,*)$ is a magma where $*$ is associative, i.e.~$(a*b)*c=a*(b*c)$ for all $a,b,c\in M$.
A \emph{monoid} $(M,*)$ is a semigroup with an unit element, i.e.~$\exists e\in M$ such that $e*a = a*e = a$ for all $a\in S$.
For a magma $(M,*)$ operation $*$ is called \emph{commutative}, if $a*b=b*a$ for all $a,b\in M$.
\end{definition}

\begin{definition}
\label{def:switching}
Let $(M,*)$ be a magma.
The generalised switching functions over M with $n$ variables ($n \in \N$) are defined as
\[ GSF(n,M):=\{\Bool^n \rightarrow M\}\]
For $i \in \{1,\dots , n\}$ define the mapping A(i) (called abstraction with respect to $*$) by
\[
\begin{array}{cccc}
A(i): & GSF(n,M) & \rightarrow & GSF(n,M) \\
      & f        & \mapsto     & A(i)(f)
\end{array}
\]
where 
\[
\begin{array}{cccl}
A(i)(f): & \Bool^n           & \rightarrow & M \\
         & (b_1,\ldots, b_n) & \mapsto     &   f(b_1,\ldots, b_{i-1},0,b_{i+1},\ldots ,b_n)*f(b_1,\ldots, b_{i-1},1,b_{i+1},\ldots ,b_n).
\end{array}
\]
If the context is clear we omit the term ``with respect to $*$''.
\end{definition}

\begin{remark}
This work centers around a natural question: Under which conditions is the following diagram commutative ($i,j\in \{1,2,\ldots, n\}$, $i\neq j$):
\begin{displaymath}
    \xymatrix{ GSF(n,M) \ar[rr]^{A(i)} \ar[d]_{A(j)} & \  & GSF(n,M) \ar[d]^{A(j)} \\
               GSF(n,M) \ar[rr]_{A(i)} & \  & GSF(n,M)}
\end{displaymath}
\end{remark}

\begin{definition}
\label{def:abstractable}
Let $n \in \N$ be arbitrary but fixed.
\begin{enumerate}
\item $f \in GSF(n,M)$ is called \emph{abstractable} 
      when it holds that $A(i)\circ A(j)(f) = A(j)\circ A(i)(f)$ $\forall i,j \in \{1,2,\ldots, n\}$.
\item $GSF(n,M)$ is called \emph{abstractable} 
      when all 
      $f\in GSF(n,M)$ are abstractable.
\end{enumerate}
\end{definition}

\begin{remark}
By definition, $GSF(1,M)$ is abstractable for every $(M,*)$.
\end{remark}

\section{The case $n=2$}
\label{sec:twoelementabstraction}
Fix a magma $(M,*)$.

\begin{lemma}
\label{2-abstraction}
Abstractability is characterised as follows:
\begin{enumerate}
  \item $f\in GSF(2,M)$ abstractable $\Leftrightarrow$ $(f(0,0)*f(1,0))*(f(0,1)*f(1,1))=(f(0,0)*f(0,1))*(f(1,0)*f(1,1))$
  \item $GSF(2,M)$ abstractable $\Leftrightarrow$ $(a*b)*(c*d)=(a*c)*(b*d)$ $\forall a,b,c,d \in M$
\end{enumerate}
\end{lemma}

\begin{proof}
The first part of the Lemma will be used in the proof of the second part.
\begin{enumerate}
  \item Two abstractions are possible and they have to be equal:
        \begin{enumerate}
          \item $A(2)\circ A(1)(f)=\ldots=(f(0,0)*f(1,0))*(f(0,1)*f(1,1))$
          \item $A(1)\circ A(2)(f)=\ldots=(f(0,0)*f(0,1))*(f(1,0)*f(1,1))$
         \end{enumerate}
  \item $\Leftarrow$ Let $GSF(2,M)$ be not abstractable $\Rightarrow$ $\exists f\in GSF(2,M)$ that is not abstractable. Define $a:=f(0,0)$, $b:=f(0,1)$, $c:=f(1,0)$, $d:=f(1,1)$.
                      Then (by part 1 of the lemma) $(a*b)*(c*d)\neq (a*c)*(b*d)$ \\
        $\Rightarrow$ Let $(a*b)*(c*d)\neq (a*c)*(b*d)$ and define $f(0.0):=a$, $f(0,1):=b$, $f(1,0):=c$ and $f(1,1):=d$. Then (by part 1) $f$ not abstractable.
\end{enumerate}

\end{proof}

\section{The case $n\in \N$}
\label{sec:bign}

\begin{lemma}
\label{lemma2a}
The following statements are equivalent ($n\geq 2$):\\
$GSF(n,M)$ abstractable $\Leftrightarrow$ $(a*b)*(c*d)=(a*c)*(b*d)$ holds for all $a,b,c,d\in M$.
\end{lemma}

\begin{proof}
$\Rightarrow$: Assume $\exists (a,b,c,d)\in M^4: (a*b)*(c*d)\neq (a*c)*(b*d)$. Define 
\[
\begin{array}{rccl}
f: & \Bool^n               &\rightarrow & M \\
   & (b_1,b_2,\ldots, b_n) & \mapsto    & \begin{cases}a & \texttt{ for } (b_1,b_2)=(0,0) \\
                                                {b} & \texttt{ for } (b_1,b_2)=(0,1) \\
                                                {c} & \texttt{ for } (b_1,b_2)=(1,0) \\
                                                {d} & \texttt{ for } (b_1,b_2)=(1,1) \end{cases}
\end{array}
\]
then $f\in GSF(n,M)$, but 
\begin{equation}
\label{eq:hinundher}
\begin{array}{rl}
     & A(1)\circ A(2)(f)(b_1,\ldots,b_n) \\
  =  & \left((\underbrace{f(0,0,b_3\ldots, b_n)}_{a}*\underbrace{f(0,1,b_3\ldots, b_n)}_{b}) * (\underbrace{f(1,0,b_3\ldots, b_n)}_{c}*\underbrace{f(1,1,b_3\ldots, b_n)}_{d}) \right)\\
\neq &  \left( (\underbrace{f(0,0,b_3\ldots, b_n)}_{a}*\underbrace{f(1,0,b_3\ldots, b_n)}_{c}) * (\underbrace{f(0,1,b_3\ldots, b_n)}_{b}*\underbrace{f(1,1,b_3\ldots, b_n)}_{d}) \right)\\
  =  & A(2)\circ A(1)(f)(b_1,\ldots,b_n),
\end{array}
\end{equation}
so $A(1)\circ A(2)(f)\neq A(2)\circ A(1) (f)$.\\
$\Leftarrow$: Assume that $GSF(n,M)$ is not abstractable. Then $\exists \ f\in GSF(n,M)$ and $i,j\in \{1,2,\ldots, n\}$, $i\neq j$ such that
$A(i)\circ A(j)(f)\neq A(i)\circ A(j)(f)$, that is $\exists \beta = (b_1,b_2,\ldots,b_n): A(i)\circ A(j)(f) (\beta) \neq A(i)\circ A(j)(f) (\beta)$.
Assume without loss of generality $i<j$.
The following function is in $GSF(2,M)$:
\[ \begin{array}{cccl}
   h: & \Bool^2 & \rightarrow & M \\
      & (x,y)   & \mapsto     & f(\beta_1,\beta_2 \ldots, \beta_{i-1}, x, \beta_{i+1},\ldots, \beta_{j-1},y,\beta_{j+1},\ldots, \beta_n)
   \end{array}.
\]
By definition $A(1)\circ A(2)(h)\neq A(2)\circ A(1) (h) \stackrel{\texttt{Lemma \ref{2-abstraction}}}{\Rightarrow} \exists a,b,c,d \in M$ with $(a*b)*(c*d)\neq (a*b)*(c*d)$.
\end{proof}

\noindent
Lemma \ref{lemma2a} 
motivates
the following definition:

\begin{definition}
$(M,*)$ is abstractable when 
$(a*b)*(c*d)= (a*c)*(b*d)$ holds for all $a,b,c,d \in M$.
\end{definition}

\begin{corollary}
$(M,*)$ abstractable $\Leftrightarrow$ $\exists n\in \N, n\geq 2$, $GSF(n,M)$ is abstractable $\Leftrightarrow$ $\forall n\in \N$ $GSF(n,M)$ is abstractable.
\end{corollary}

\begin{proof}
Lemma \ref{lemma2a} shows that abstractability does not depend on $n$ ($n\geq 2$).
\end{proof}

\section{Associativity, commutativity and abstractability}
\label{sec:assoc}
The following two lemmata show the connection between associativity, commutativity and abstractability.
\begin{lemma}
\label{fundamental1}
  Let $(M,*)$ be a magma, then it holds that\\
  ( $*$ commutative $\wedge$ $*$ associative ) $\Rightarrow$ $(M,*)$ abstractable.
\end{lemma}
\begin{proof}
This is the trivial implication one immediately sees from $(x_1*x_2)*(x_3*x_4)\stackrel{!}{=}(x_1*x_3)*(x_2*x_4)$. 
As $*$ is associative, all parentheses can be omitted, so it remains 
to swap $x_2$ and $x_3$ to get equality. But this is possible due to commutativity. \\
\end{proof}

\begin{lemma}
\label{fundamental2komm}
  Let $(M,*)$ be a magma, 
  then it holds that\\
  ( $(M,*)$ abstractable $\wedge$ $\exists$ left and right unit elements $e_l, e_r\in M$ ) $\Rightarrow$ ($*$ commutative).\\
\end{lemma}

\begin{proof}
Assuming abstractability, $(a*b)*(c*d)=(a*c)*(b*d)$ holds for all $a,b,c,d \in M$.
Setting $a=e_l$ $d=e_r$ it must also hold that $b * c = c * b$ for arbitrary $b$, $c$ (commutativity).
\end{proof}

\begin{lemma}
\label{fundamental2}
  Let $(M,*)$ be a magma, 
  then it holds that\\
  ( $(M,*)$ abstractable $\wedge$ $\exists$ unit element $e \in M$ ) $\Rightarrow$ ($*$ commutative $\wedge$ $*$ associative).\\
\end{lemma}

\begin{proof}
Assume that $e*m=m=m*e$ is the unit element in $(M,*)$. Assuming abstractability, $(a*b)*(c*d)=(a*c)*(b*d)$ holds for all $a,b,c,d \in M$. 
Then (setting $b=e$) one has 
$a*(c*d)\stackrel{\texttt{unit}}{=}(a*e)*(c*d)\stackrel{\texttt{abstractable}}{=}(a*c)*(e*d)\stackrel{\texttt{unit}}{=}(a*c)*d$ (associativity).
The commutativity follows by Lemma \ref{fundamental2komm}.
\end{proof}

\begin{remark}
The converse of Lemma \ref{fundamental2} is not true: For $(\N\setminus \{0\},+)$, $+$ is commutative and associative, 
so by Lemma \ref{fundamental1} abstractable.
But as $0\notin \N\setminus \{0\}$, no unit element is present.
\end{remark}

\noindent
In the case of a magma with unit element, one can immediately conclude:

\begin{corollary}
\label{cor:zentral}
For a magma $(M,*)$ with unit element $e$ it holds that\\
$(M,*)$ abstractable $\Leftrightarrow$ ($*$ commutative $\wedge$ $*$ associative),\\
in particular if there is an abstractable magma $(M,*)$ where $*$ is not associative or not commutative,
it cannot have a unit element.
\end{corollary}

%

\noindent The following box shows the contributions of this paper. \\
\fbox{\begin{minipage}{0.9\linewidth}
The following statements hold for the given algebraic structures:
\begin{enumerate}
  \item $\texttt{monoid } (M,*) \Rightarrow [ (*$ commutative) $\Leftrightarrow (M,*)$ abstractable ]
  \item $\texttt{semigroup } (M,*)$
      \begin{enumerate}
        \item $*$ commutative $\Rightarrow (M,*)$ abstractable
        \item $*$ non-commutative:
                 \begin{enumerate}
                   \item $\exists (M,*):$\ $(M,*)$ abstractable
                   \item $\exists (M,*):$\ $(M,*)$ not abstractable
                 \end{enumerate}
      \end{enumerate}
  \item $\exists \texttt{ magma } (M,*):$\ ($*$ non-associative) $\wedge (M,*)$ abstractable
\end{enumerate}\end{minipage}}\\
\begin{proof}
Statement 1 is Corollary \ref{cor:zentral}. Statement 2(a) is Lemma \ref{fundamental1}.
Statement 2(b)i is proven in Sec.~\ref{sec:positiveexample}. Statement 2(b)ii is treated in Sec.~\ref{sec:negativeexample}. 
Statement 3 is justified by the examples in Sec.~\ref{sec:counterexample} (one commutative and one non-commutative example).
\end{proof}

\section{The case of a monoid}
\label{sec:monoid}
This section gives some examples for monoids.

\subsection{Non-abstractable examples}
\subsubsection{Permutation group}
\label{sec:permutations}
It is well-known that the permutations of 3 elements define an associative but not commutative group with composition of permutations as operation.
The group is denoted by $(S_3,\circ)$. As a group is, in particular, also a monoid, Corollary \ref{cor:zentral} shows that $(S_3,\circ)$ is not abstractable.

\subsubsection{Matrix multiplication}
It is well-known that matrix multiplication is associative but not commutative.
As there is also a unit element, $Mat(2\times 2, \R)$ forms a non-commutative monoid. Therefore by Corollary \ref{cor:zentral} it is clear that 
$(Mat(2\times 2, \R),\cdot)$ is not abstractable.

\subsection{Abstractable examples}

Here one can take the usual examples $(\Z,+)$, $(\R,\cdot)$, etc.~as they form commutative monoids.

\section{The case of a magma}
\label{sec:counterexample}
The aim of this section is to show that neither commutativity nor associativity is required for $(M,*)$ to be abstractable.
Especially these examples show that associativity is not necessary and therefore disprove the statements
``The key point here is associativity. If operator $*$ does not have this property, then the result of generic abstraction depends on the order of applying $*$ $\ldots$''
(\cite{bahar:97}, p.~179),
``The operator $*$ must be associative as the order in which the co-factors and the variables
are chosen shall not influence the outcome of the operation.'' (\cite{kuntz:06}, p.~39) and
``It now becomes obvious that associativity of the operator $*$ is required in order to ensure
that the ordering of the variables and the order in which the cofactors are chosen do not 
influence the outcome.'' (\cite{Siegle:Habil}, p. 40/41).

\subsection{Not associative and not commutative magma (finite)}
Given the two-element set $\{0,1\}$ with the magma:
{\small\[
  \begin{tabular}{|c||c|c|} \hline
  * & 0 & 1 \\ \hline \hline
  0 & 1 & 0 \\ \hline
  1 & 1 & 0 \\ \hline
  \end{tabular}
\]}
Looking at the composition table it is clear that it is not commutative $(0*1\neq 1*0)$ and not associative $\underbrace{(0*0)}_1*0\neq 0*\underbrace{(0*0)}_1$.
But still $(\{0,1\},*)$ is abstractable: Defining $\bar{0}:=1$, $\bar{1}:=0$ it holds
that $x*y=\bar{y}$, therefore $\underbrace{(x_1 * x_2)}_{\bar{x}_2} * \underbrace{(x_3 * x_4)}_{\bar{x}_4}=x_4$.
Therefore it is trivially abstractable.

\subsection{Not associative and not commutative magma (infinite, by Greither)}
The set $\Z:=\{\ldots,-2,-1,0,1,2,\ldots\}$ with the operation $-$. The subtraction is not commutative ($2-3=-1$, but $3-2=1$), not associative
$(1-(2-3)=2\neq (1-2)-3=-4)$, but is abstractable as $(x_1-x_2)-(x_3-x_4)=x_1-x_2-x_3+x_4=(x_1-x_3)-(x_2-x_4)$. 




\subsection{Not associative but commutative magma}
\label{example:non-trivial}
This example is a non-associative (but commutative) mapping $*$ on a set $\{a,b,c,d\}$ where $(\{a,b,c,d\},*)$ is abstractable.
Given a magma $(\{a,b,c,d\},*)$ with the operation $*$ defined as:
{\small
\[
  \begin{tabular}{|c||c|c|c|c|} \hline
  * & a & b & c & d \\ \hline \hline
  a & a & c & b & d \\ \hline
  b & c & a & d & b \\ \hline
  c & b & d & a & c \\ \hline
  d & d & b & c & a \\ \hline
  \end{tabular}
\]}
Then $*$ is not associative, as $(a*b)*c=a$, but $a*(b*c)=d$. 
It can be seen that $(M,*)$ is (non-trivially) abstractable by checking all possible combinations (e.g.~by a computer program) or by the following verification
that exploits symmetries (note that $*$ is commutative):
%
\begin{enumerate}\item Four different values: $(x_1*x_2)*(x_3*x_4)=a$ \\
                         Therefore one checks \\
                         {\small
                         \begin{tabular}{|c|c|c|c|c|c|} \hline
                         $x_1$ & $x_2$ & $x_3$ & $x_4$ & $x_1*x_2$ & $x_3*x_4$  \\ \hline \hline
                         a     & b     & c   & d   &  c        & c          \\ \hline
                         a     & c     & b   & d   &  b        & b          \\ \hline
                         a     & d     & b   & c   &  d        & d          \\ \hline
                         \end{tabular}} \\
                         and the remaining cases follow by commutativity.
                         
                   \item Three different values: $\underbrace{(x_1*x_1)}_a*(x_2*x_3)=(x_1*x_2)*(x_1*x_3)$
                         Here one has to check a little bit more: \\
                         {\small
                         \begin{tabular}{|c|c|c|c|l|} \hline
                         $x_2$ & $x_3$ & $a*(x_2*x_3)$ & $x_1$ & $(x_1*x_2)*(x_1*x_3)$ \\ \hline \hline
                         a     & b     & b             & c     & b * d = b  \\ \cline{4-5}
                               &       &               & d     & d * b = b  \\ \hline
                         \end{tabular}} \\
                         One checks similarly the cases $$(x_2,x_3,x_1)=\{(a,c,b/d), (a,d,b/c), (b,c,a/d), (b,d,a/c), (c,d,a/b)\}.$$ The remaining cases follow by commutativity.
                    \item The case of three equal values is clear by commutativity.
                    \item For four equal values there is nothing to show.
\end{enumerate}

\section{The case of a semigroup}
\label{sec:semigroup}
Let $(M,*)$ be a semigroup (* is associative, but not necessarily commutative and no unit element is required).

\subsection{Non-abstractable example (Tamura)}
\label{sec:negativeexample}

\label{sec:tamura}
This semigroup $(\{a,b,c,d\},*)$ is one of the examples presented in \cite{takayuki:56}. The operation is defined as
{\small
\[
  \begin{tabular}{|c||c|c|c|c|} \hline
  * & a & b & c & d \\ \hline \hline
  a & a & a & a & a \\ \hline
  b & b & b & b & b \\ \hline
  c & c & c & c & c \\ \hline
  d & a & a & b & a \\ \hline
  \end{tabular}.
\]}
The reader may convince himself that the operation is indeed associative. From the table it is clear that $*$ is not commutative.
Further, $(\{a,b,c,d\},*)$ cannot be a monoid, as e.g.~there is no element $e$ such that $e*d=d$ ($d$ is not even in the image of the operation).
Now $(d*a)*(c*b)=a\neq (d*c)*(a*b)=b$ and therefore $(\{a,b,c,d\},*)$ is not abstractable.

\subsection{Abstractable examples}
\label{sec:positiveexample}


\label{sec:projections}
There are abstractable semigroups $(M,*)$ where the operation $*$ is not commutative.
\begin{definition}
A \emph{trivial associative non-commutative semigroup} is a set $M$ together with a (left or right) projection $*:M\times M\rightarrow M$.
\end{definition}

\begin{lemma}
\label{projlemma}
If $(M,*)$ is a trivial associative non-commutative semigroup, then $(M,*)$ is trivially abstractable.
\end{lemma}

\begin{proof}
Clear by definition.
\end{proof}

%

\section{Operations on $\R$}
\label{sec:onR}

In this section we show how the results of the previous sections can be applied to the real numbers.

\subsection{Structure transport}
We give a construction for a structure transport from semigroups to sets.
Thereby, for example, an associative but non-commutative operation on $\R$ (modulo a certain equivalence relation) can be defined.
\begin{lemma}
\label{lemma1}
Given a set $S$, a non-commutative semigroup $(M,*)$ (i.e.~$\exists a,b\in M: a* b \neq b* a$)
and a surjective mapping $f:S\rightarrow M$. Then there exists an injective mapping $g:M\rightarrow S$, such that $f\circ g= id$ (identity)
and a non-commutative associative operation on $\nicefrac{S}{\sim}$ can be defined by
\begin{eqnarray*}
*^{*}:& \nicefrac{S}{\sim}\times \nicefrac{S}{\sim} & \rightarrow  \nicefrac{S}{\sim} \\
      & (a,b)     & \mapsto      g(f(a) * f(b)).
\end{eqnarray*}
Here $\sim$ is the equivalence relation defined by $a\sim b:\Leftrightarrow g(f(a))=g(f(b))$.
\end{lemma}

\begin{proof} The mapping $g$ is well-known to exist. It remains to show that $(\nicefrac{S}{\sim},*^{*})$ is a non-commutative semigroup: \\
{\bf Associativity:} $a*^{*}(b *^{*} c) = g\left(f(a) * \underbrace{f( g}_{id}(f(b) * f(c)) )\right) = g\left(f(a) * (f(b) * f(c)) \right)$.
Analogously one has $(a*^{*}b) *^{*} c = \ldots = g\left( ( f(a) * f(b) ) * f(c) \right)$. So because of the associativity of $*$
the associativity is proven. \\
{\bf Non-commutative:} As $(M,*)$ is non-commutative, it follows that $\exists m_1,m_2\in M : m_1* m_2 \neq m_2 * m_1$.
As $f$ is surjective there exist $s_1, s_2$ such that $f(s_1)=m_1$, $f(s_2)=m_2$. Therefore $f(s_1)* f(s_2)\neq f(s_2)* f(s_1)$.
As $g$ is injective it follows that $s_1*^* s_2 \neq s_2 *^* s_1$.
\end{proof}

\subsection{Associative operations where $(\nicefrac{\R}{\sim},*)$ is not abstractable}

\subsubsection{Permutation group}
For example choose the surjection 
\begin{eqnarray*}\R \rightarrow S_3 \\
                  x \mapsto \sigma_{\lceil x \rceil\mod 6}
\end{eqnarray*}
and an appropriate mapping $S_3\rightarrow \R$.
That construction leads by Lemma \ref{lemma1} to a semigroup (that is already a monoid). By Corollary \ref{cor:zentral} it is clear that $(\nicefrac{\R}{\sim},*)$ is not abstractable.

\subsubsection{$\R^2$ using matrix multiplication}
Using the following mapping, $\R^2$ can be endowed with an associative, but non-commutative operation:
\begin{eqnarray*}
m: & \R^2 \rightarrow & Mat(2\times 2, \R) \\
   & (a,b) \mapsto   &    \begin{cases} \left(\begin{array}{cc} 1 & 0 \\ 0 & 1 \end{array}\right) & \texttt{ for } a=b=0 \\
                                        \left(\begin{array}{cc} a & a \\ b & b \end{array}\right) & \texttt{ otherwise }
                          \end{cases}
\end{eqnarray*}
With the usual matrix multiplication it is clear that the resulting matrices are again of this special form, so it is indeed an operation.
In this way $\R^2$ is equipped with a non-commutative monoid structure (cf.~Lemma \ref{lemma1}). 
The mapping $m: \R^2 \setminus \{(0,0)\}\rightarrow Mat(2\times 2, \R)$ is continuous. 
Note that no two elements of $\R^2$ are identified by $\sim$. The monoid $(\R^2,*)$ is not abstractable.
%
%
%
%
This monoid can be used to define an operation on $\R$ modulo some equivalence relation. 
A mapping $f:\R \rightarrow \R^2$ can be defined by an appropriate space filling curve (e.g.~Peano space filling curve which is surjective and smooth, but \emph{not} injective) and an isomorphism 
$(0,1)\simeq \R$ (e.g.~$\tan(\pi(x-\frac{1}{2}))$). So one has the (surjective) mapping $\R \simeq (0,1) \stackrel{\texttt{Peano}}{\rightarrow} (0,1)^2 \simeq \R^2$
that can be used for the structure transport by Lemma ~\ref{lemma1} (choosing an appropriate mapping $\R^2\rightarrow \R$). Again, $(\nicefrac{\R}{\sim},*)$ is not abstractable.

\subsection{Associative operations where $(\R,*)$ is abstractable}

\subsubsection{Projection functions}
The projection functions 
\begin{eqnarray*}
l: & \R \times \R   \rightarrow \R \\
   & (a,b)          \mapsto a
\end{eqnarray*}
and 
\begin{eqnarray*}
r: & \R \times \R   \rightarrow \R \\
   & (a,b)          \mapsto b
\end{eqnarray*}
can be easily shown to be non-commutative. By Lemma \ref{projlemma} $(\R,l)$ and $(\R,r)$ are abstractable.

\begin{remark}  
Among all affine mappings $\R^2\rightarrow \R$, the only operations that are associative and not commutative are the projection functions:
Assume $t(x,y):=ax+by+c$ with $a,b,c\in \R$. Then, because of non-commutativity, one has $a\neq b$. For associativity it must hold that
$t(x,t(y,z))=t(t(x,y),z)$ and therefore the terms
\begin{eqnarray*}
  t(x,t(y,z))= & t(x,ay+bz+c)=ax+b(ay+bz+c)+c=ax+aby+b^2z+bc+c \\
  t(t(x,y),z)= & t(ax+by+c,z)=a(ax+by+c)+bz+c=a^2x+aby+bz+ac+c
\end{eqnarray*}
must coincide for all $x$, $y$, $z\in \R$. 
Comparing coefficients leads to $a^2=a$, $b^2=b$ and therefore $a,b\in\{0,1\}$.
From $bc+c=ac+c$ and $a\neq b$ it follows that $c=0$. Therefore only the projections remain.
\end{remark}
%

\subsubsection{Other continuous examples} 
There are 
associative non-commutative
{\bf{continuous}} operations
    $$l: \R \times \R \longrightarrow \R $$
where $(\R,l)$ is abstractable.
Exemplary construction:

\begin{enumerate}
\item 
   $M:= ]4, \infty [$ and 
   
    $$\begin{array}{rcl}h: & M \times M & \longrightarrow M  \\
                            & (   x,y)    & \mapsto \min(\min(x+1,16) \cdot y , 64)\end{array}$$
   
\item 
   $h$ is well-defined:\\ 
   
   $x,y\in M \Rightarrow \min(x+1,16) > 5 \Rightarrow (\min(x+1,16) \cdot y) > 5\cdot 4 = 20 \Rightarrow $\\
   $\min(\min(x+1,16) \cdot y , 64) >20$.

   So $$ h(x,y) > 20 \,\,\forall\,\, (x,y) \in M \times M$$
   or 
   $20$ is a lower bound of $h(M \times M)$.

\item 
 $h$ is associative:\\
   $x,y,z \in M$.
   \begin{itemize}
   \item
   Knowing $ h(x,y) > 20 $  and $z>4$:\\
   $
   h( (h(x,y) , z) = 
     \min( \min(h(x,y)+1 ,16) \cdot z , 64) =  
     \min( (16 \cdot z) , 64) = 64
   $

   \item
   Knowing $ h(y,z) > 20 $  and $x>4 \Rightarrow$ 
   $ \min(x+1,16) \cdot h(y,z) > 64$, so \\
   $
   h(x, h(y,z)) = \min(\min(x+1,16) \cdot h(y,z) , 64) = 64
   $
   \end{itemize}

   $$h( (h(x,y) , z) = h(x, h(y,z)) = 64 \,\, \forall\,\, x,y,z \in M$$

\item 
 $h$ is not commutative:\\ $h(5,4) = 24$ and $h(4,5) = 25$.

\item 
 $h$ is continuous as a composition of continuous mappings 

\item 
   There exists a homeomorphism $f: \R \longrightarrow M$ (e.g. $x\mapsto
   \exp(x)+4$).\\
   Using the construction given in the proof of Lemma \ref{lemma1} one gets: 

   $$\begin{array}{rcl}l: & \R \times \R & \longrightarrow \R \\
                           &             (a,b)               & \mapsto f^{-1}( h(f(a) , f(b)) )\end{array} $$

   $l$ is continuous, associative and non-commutative. 
   Note that nothing is identified by $\sim$ for this construction.

\end{enumerate}
So, as in the abstraction formula always at least two consecutive operations have to be performed, the result of
any abstraction is equal to $64$, i.e.~constant (which can also be seen as some kind of triviality).

\section{Conclusion}
\label{sec:conclusion}
We have corrected and extended statements made in \cite{bahar:97,kuntz:06,Siegle:Habil}. 
It is now clear that not \emph{associativity} is the crucial point, but
one has to differentiate. In the case of a monoid, \emph{commutativity} is crucial (here: necessary and sufficient) for the independence 
of the abstraction operation of the 
order of the abstracted variables. For the cases of magmas, semigroups and monoids a couple of examples show 
how the abstraction behaves for a set of abstracted variables. 
The paper shows that the well-definedness of the abstraction operation is equivalent to the operation being an abstractable magma.
Basic examples of abstractable magmas over the real numbers (not known to the MTBDD community prior to this paper) 
are subtraction and division.
%
%

{\bf Acknowledgements:} Special thanks to Markus Siegle: In an exercise to his lecture on distributed systems the question arose
that was the spark to this paper. Further he first gave a simplified version of the proof of Lemma \ref{lemma2a} that showed that the preimage does 
not have to be a Boolean algebra, a two element set is sufficient. We would also like to thank Cornelius Greither who enormously helped to clarify
the statements (especially Lemma \ref{fundamental2}) and gave another nice example.
We would further like to thank Alexander Gouberman for the fruitful mathematical discussions.

\bibliographystyle{spmpsci}      
\bibliography{springer}   

\begin{thebibliography}{1}
\providecommand{\url}[1]{{#1}}
\providecommand{\urlprefix}{URL }
\expandafter\ifx\csname urlstyle\endcsname\relax
  \providecommand{\doi}[1]{DOI~\discretionary{}{}{}#1}\else
  \providecommand{\doi}{DOI~\discretionary{}{}{}\begingroup
  \urlstyle{rm}\Url}\fi

\bibitem{bahar:97}
Bahar, R., Frohm, E., C.M.Ganoa, Hachtel, G., Macii, E., Pardo, A., Somenzi,
  F.: {Algebraic Decision Diagrams and their Applications}.
\newblock Formal Methods in System Design \textbf{10}(2/3), 171--206 (1997)

\bibitem{jezek:83}
Jezek, J., Kepka, T.: Medial Groupoids.
\newblock Rozpravy Ceskoslovenske Akademie Ved, Rada matematickych a prirodnch
  ved, Rocnik 93, Sesit 2, Prague (1983)

\bibitem{kuntz:06}
Kuntz, M.: {Symbolic Semantics and Verification of Stochastic Process
  Algebras}.
\newblock Dissertation, {Institut f\"{u}r Informatik der FAU Erlangen} (2006)

\bibitem{parker:prism}
{PRISM website}.
\newblock {\url{http://www.prismmodelchecker.org/}}

\bibitem{shcher:05}
Shcherbacov, V.: {On the structure of finite medial quasigroups}.
\newblock Bul. Acad. \c Stiin\c te Repub. Mold. Matematica \textbf{1}(47),
  11--18 (2005)

\bibitem{Siegle:Habil}
Siegle, M.: {Behaviour analysis of communication systems: Compositional
  modelling, compact representation and analysis of performability properties}.
\newblock Shaker-Verlag (2002)

\bibitem{strecker:74}
Strecker, R.: {\"Uber entropische Gruppoide}.
\newblock Mathematische Nachrichten \textbf{64}(1), 363--371 (1974)

\bibitem{takayuki:56}
Tamura, T.: {The Theory of Construction of Finite Semigroups I}.
\newblock Osaka Math. Journal \textbf{8}(2) (1956)

\end{thebibliography}

\end{document}